\documentclass{ecai} 
\usepackage{latexsym}
\usepackage{amssymb}
\usepackage{subcaption}
\usepackage{rotating}
\usepackage{amsmath}
\usepackage{amsthm}
\usepackage{hyperref}
\usepackage{url}
\usepackage{booktabs}
\usepackage{xcolor}
\usepackage{nicefrac}
\usepackage{enumitem}
\usepackage{graphicx}
\usepackage{adjustbox}
\newcommand{\BibTeX}{B\kern-.05em{\sc i\kern-.025em b}\kern-.08em\TeX}
\newtheorem{theorem}{Theorem}

\newtheorem{observation}[theorem]{Observation}

\newtheorem{definition}{Definition}
\newtheorem{remark}[definition]{Remark}
\newtheorem{example}[definition]{Example}

\begin{document}
\begin{frontmatter}

\title{Algorithms for Collaborative Harmonization}

\author[1]{\fnms{Eyal} \snm{Briman}}
\author[1]{\fnms{Eyal} \snm{Leizerovich}}
\author[1]{\fnms{Nimrod} \snm{Talmon}}

\address[1]{Ben Gurion University of the Negev}

\begin{abstract}
We consider a specific scenario of text aggregation, in the realm of musical harmonization. Musical harmonization shares similarities with text aggregation, however the language of harmony is more structured than general text. Concretely, given a set of harmonization suggestions for a given musical melody, our interest lies in devising aggregation algorithms that yield an harmonization sequence that satisfies the following two key criteria:
(1) an effective representation of the collective suggestions;
and (2) an harmonization that is musically coherent.
We present different algorithms for the aggregation of harmonies given by a group of agents and analyze their complexities. The results indicate that the Kemeny and plurality-based algorithms are most effective in assessing representation and maintaining musical coherence.
\end{abstract}
\end{frontmatter}

\section{Introduction} 
Social choice theory provides aggregation algorithms that facilitate collaborative creation of diverse outputs within agent communities: e.g., single-winner elections, multi-winner elections~\cite{arrow2010handbook}, and participatory budgeting~\cite{skowronparticipatory,aziz2021participatory}.
No good solutions, however, exist for the aggregation of preferences regarding the collaborative creation of text documents.
The Following work serves as a milestone towards this process of collaborative text writing.
We consider a musical melody and a population of agents, each of which is suggesting a different harmonic sequence (equivalently, an harmonization) for the melody; and our aim is to aggregate those suggestions to come to a single harmonization. Our pursuit is to discover algorithms that strike a balance between respecting the agent community's preferences and crafting aggregated sequence of harmonies that are likely to appear according to prior knowledge of chord sequences translated to a 2-gram. 

To this end, we first model the problem of collaborative harmonization and then explore various specially-crafted aggregation algorithms in pursuit of this goal. We then report on computer-based simulations performed on real-world data and generated data.\footnote{In this work, we present an aggregation method for harmonies without addressing the melody. We assume for simplicity needs, that the user is responsible for providing a valid harmony to accompany a melody, and we solely offer the aggregation methods.}

\paragraph{Paper Structure}
After reviewing related work, we provide musical preliminaries~\ref{section:preliminaries}).
This groundwork will facilitate the formal delineation of the social choice framework for collaborative harmonization in Section~\ref{section:formal model}.
Proceeding from there, we delineate various approaches and algorithms addressing the problem in Section~\ref{section:aaa} and conduct a comprehensive analysis of the computational aspects associated with each approach in Section~\ref{section:complexity}.
Concluding the substantive portions of the paper, we execute simulations to evaluate the algorithms and approaches in Section~\ref{section:simulations}, presenting and scrutinizing the results in detail.

\paragraph{Related Work} 
We mention related work from both the social choice literature and chord sequence generation.
First, we mention the work on general aggregation in metric spaces~\cite{bulteau2018aggregation,zvi2021iterative}, that also includes suggestions on how to perform text aggregation. Another model that is relevant is that of \emph{multiple attribute list aggregation}~\cite{briman2023multiple}, in which a sequence of elements (not necessarily text characters) are the output of the social choice instance. We also wish to mention judgment aggregation~\cite{grossi2022judgment}, which corresponds to a very general social choice setting and for which some of our algorithms are related, and get inspiration from the general work on aggregating preferences under constraints in formatted ways~\cite{li2015aggregating}
Next we mention works on decision-theoretic
planning techniques into automatic harmony generation and chords progression generation based on stochastic processes  ~\cite{paiement2005probabilistic,clement1998learning,yi2007automatic}. Our work makes assumptions for simplicity needs that differ from the classic chord progression problems, as we are interested in aggregation of chords preferences of agents, with respecting the probability of a chord progression to appear (based on pre-trained 2-gram model), but with no explicit use of the melody it self-- we assume the user is responsible of giving a valid harmony to a melody, and we offer an aggregation method.

\section{Musical Preliminaries}\label{section:preliminaries}

While the focus of the paper is rather on its social choice aspects, it is nevertheless essential to provide a foundational understanding of the musical elements we will be working with. Music, like other complex systems, can be viewed mathematically. At its core, music is composed of three fundamental components: \textit{rhythm}, \textit{melody} and, \textit{harmony}. Below, to make the preliminaries accessible also to readers without a background in music theory, we break these components using a more abstract, mathematical framework~\cite{levine2011jazz,rohrmeier2020syntax}.

\paragraph{Rhythm}
Rhythm in music is the organized arrangement of sound and silence within time. A melody consists of beats, grouped into sections called bars or measures, similar to paragraphs in writing. Each bar, determined by a time signature, contains a specific number of beats, indicating the beat count and the type of note representing one beat. This rhythmic structure forms the pulse of a song, guiding musicians and engaging listeners with its rhythmic pattern.

\paragraph{Melody}
The melody, in its simplest form, can be seen as a sequence of notes played one after the other. Each note can be represented using an "alphabet", which consists of a finite set of symbols. These symbols are include letters such as C, D, E, and so on (there is, indeed, a different, equivalent "alphabet" in which the symbols are Do, Re, Mi, which we, however, will not use here). These symbols represent different pitches or musical notes.
E.g., consider the melody of the well-known nursery rhyme "Twinkle, Twinkle, Little Star." It can be broken down into a sequence of note symbols: C-C-G-G-A-A-G. Note how the arrangement of these symbols forms the song's melody.

\paragraph{Harmony}
As we are interested in collaborative harmonization, the harmony will be the crucial aspect of the music that we will be concentrated on.
In particular, in this paper we focus our investigation on the harmony; and take a close look at its rules and grammar.

Acoustically, harmony complements the melody by introducing a layer of complexity to music. It focuses on the simultaneous sounding of multiple notes, forming so-called chords; the notes of these chords are played in parallel (simultaneously) to the notes of the main melody. Continuing our mathematical framework, chords can be considered as combinations of note symbols. In our "musical alphabet", a chord like ``C-E-G'' consists of symbols representing individual notes played together (in addition to the main melody). The structure within these chord sequences arises from the specific combination of notes and their temporal arrangement.

The relationship between harmony and melody is fundamental to music. Melodies are often played over a backdrop of harmonies, creating rich and emotionally expressive compositions. The choice of chords and their sequence can significantly influence the mood and feel of a musical piece, much like how words and their arrangement in a language can convey meaning and tone. For example, playing a melody over a \textit{C major} chord (containing C, E, and G) can result in a different emotional tone than playing the same melody over an \textit{A minor} chord (containing C, E, and A).

\begin{figure}[t]
    \centering
    \begin{subfigure}{10cm}
        \centering
        \includegraphics[width=7cm]{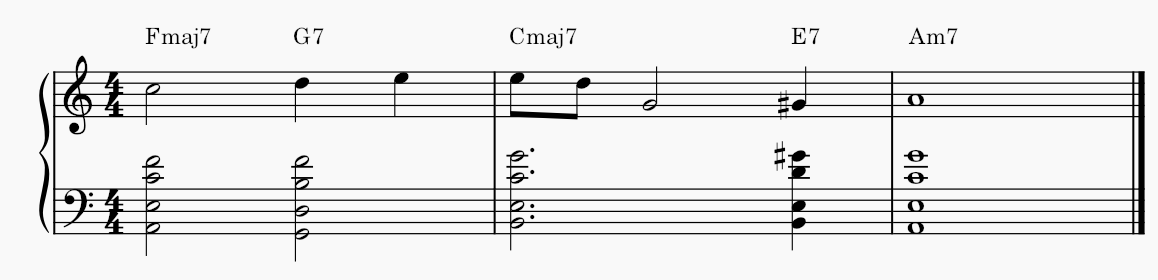}
        \caption{A simple harmonization of a melody.}
        \label{simple_harmony}
    \end{subfigure}
    
     \vspace{1cm} 
    
    \begin{subfigure}{10cm}
        \centering
        \includegraphics[width=7cm]{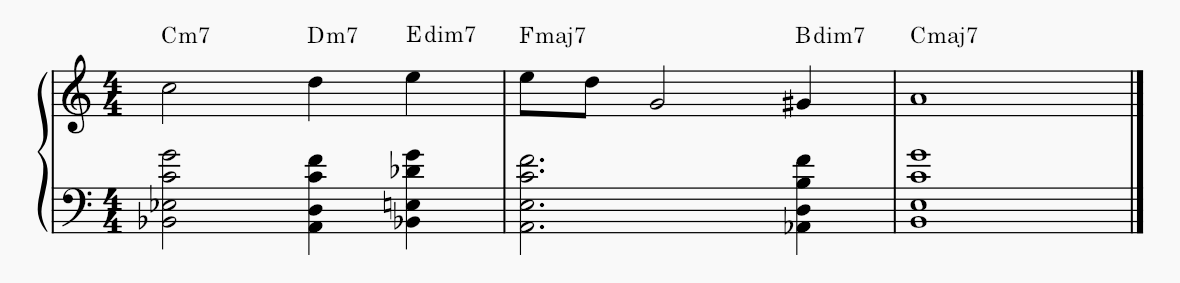}
        \caption{A more complex harmonization of the same melody.}
        \label{complex_harmony}
    \end{subfigure}
    \vspace{1cm}
    \caption{Comparison of different harmonizations.\ \\ \ \\}
    \label{fig:both_harmonies}
\end{figure}

\begin{example}
Sub-figure~\ref{simple_harmony} and sub-figure~\ref{complex_harmony} present the same melody on their upper staves, accompanied by different harmony symbols like Fmaj7, G7, Cm7, and E7. However, the bottom staves in both sub-figures depict the differing harmonic interpretations through the specific note collections associated with the harmony symbols shown on the upper staves.
\end{example}
\begin{figure}[t]
\centering
\fbox{\begin{minipage}{0.45\textwidth}
\footnotesize    \begin{itemize}[label=]
        \item CMaj7, Cm7, CmMaj7, C7, CdimMaj7, Cdim7, Cm7b5, Cm6, C+7, C+maj7,
        \item DbMaj7, Dbm7, DbmMaj7, Db7, DbdimMaj7, Dbdim7, Dbm7b5, Dbm6, Db+7, Db+maj7,
        \item DMaj7, Dm7, DmMaj7, D7, DdimMaj7, Ddim7, Dm7b5, Dm6, D+7,D+maj7,
        \item Ebmaj7, Ebm7, EbmMaj7, Eb7, EbdimMaj7, Ebdim7, Ebm7b5, Ebm6, Eb+7, Eb+maj7,
         \item Emaj7, Em7, EmMaj7, E7, EdimMaj7, Edim7, Em7b5, Em6, E+7,E+maj7,
        \item FMaj7, Fm7, FmMaj7, F7, FdimMaj7, Fdim7, Fm7b5, Fm6, F+7, F+maj7, 
        \item Gbmaj7, Gbm7, GbmMaj7, Gb7, GbdimMaj7, Gbdim7, Gbm7b5, Gbm6, Gb+7, Gb+maj7,
        \item GMaj7, Gm7, GmMaj7, G7, GdimMaj7, Gdim7, Gm7b5, Gm6, G+7, G+maj7,
        \item AbMaj7, Abm7, AbmMaj7, Ab7, AbdimMaj7, Abdim7, Abm7b5, Abm6, Ab+7, Ab+maj7,
        \item AMaj7, Am7, AmMaj7, A7, AdimMaj7, Adim7, Am7b5, Am6, A+7, A+maj7,
        \item BbMaj7, Bbm7, BbmMaj7, Bb7, BbdimMaj7, Bbdim7, Bbm7b5, Bbm6, Bb+7, Bb+maj7,
        \item BMaj7, Bm7, BmMaj7, B7, BdimMaj7, Bdim7, Bm7b5, Bm6, B+7, B+maj7.        
    \end{itemize}
\end{minipage}}
\caption{The alphabet of chords used in this study.}
\label{figure:alphabet}
\end{figure}

\subsection{A Basic Grammar of Harmony}\label{Grammar}

Our main motivation for studying the collaborative harmonization problem is that harmony can be naturally represented as text (where each chords corresponds to its character) and, more importantly, it has some general structured grammar that is more basic than general text documents.
In what follows we provide a simplified view of the grammar of harmony that is useful for our purposes.

In this paper, we a consider finite set of chords we address as "alphabet of chords" as shown in Figure~\ref{figure:alphabet} (this is not a complete list of all chords used in western music, though, but it is sufficient for the purposes of our study).
Correspondingly, an harmonic sequence is simply a sequence of characters from that alphabet of chords.


\paragraph{Chord Similarity -- Spatial Grammatic Aspect}
Considering the chords shown in Figure~\ref{figure:alphabet}, it is important to realized that, acoustically, certain chords are more similar to others. 
Essentially, this is somehow similar to the fact that certain \textbf{words} are more similar to other words (e.g., bicycle and bike are quite similar to each other, while apple and spaceship are perhaps more semantically distant).

Correspondingly, it is useful to consider a distance metric to quantify chord similarity.
Below we discuss two options.
\begin{itemize}

\item \textbf{Jaccard Distance}: Recalling that each chord (internally) contains few notes, this metric measures dissimilarity by comparing the intersection size between the notes of two chords (treated as sets of 4 notes) to the size of their union. It accounts for overlapping similar-pitch notes, providing a similarity measure that ranges from 0 (no similarity) to 1 (identical chords).

Concretely, the Jaccard distance between two chords $A$ and $B$ is
\begin{align}
    d_{\text{Jaccard}}(A,B)=1- \frac{|A\cap B|}{|A\cup B|}\ ,
\end{align}
where $|A \cap B|$ denotes the number of common notes between chords $A$ and $B$, and $|A \cup B|$ represents the total number of unique notes in both chords.

\begin{example}
In Figure~\ref{figure:two chords}, we can see that CMaj7 (C E G B) and FMaj7 (F A C E) have two common notes, namely C and E, and 6 unique notes - B, C, E, F, G, A . So their Jaccard distance is~$\frac{2}{3}$.
\end{example}

\begin{figure}[t]
    \centering
    \includegraphics[width=4cm]{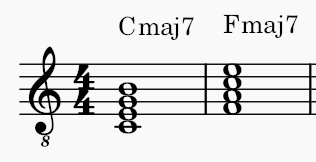}
    \caption{Two Chords: CMaj7 and FMaj7.}
    \label{figure:two chords}
\end{figure}

\item \textbf{Tonal Distance}: Besides the simple note-counting process that unerpins the Jaccard distance, there are other psychoacoustic features that affect chord similarity. The tonal distance relates to the \emph{acoustic-functional relationship} between chords: it assesses the harmonic function of chords, rather than focusing solely on the pitch content~\cite{rohrmeier2011towards,harmonyBook}.

\end{itemize}
As the concept of  \emph{acoustic-functional relationship} between chords is challenging to model directly, we have chosen the Jaccard distance as our metric for this work although it is a less common ''chord-distance'' in practice, as we are interested in the aggregation process it self. We speak about this decision more in section~\ref{section:outlook} and recommend investigating other distance metrics for future work.

\paragraph{Chord Progression -- Temporal Grammatical Aspect}
Besides the similarity (and interchangeability) of chords, there is also importance to chord progression -- how chords in a sequence relate to each other.
Essentially, this is somehow similar to the fact that certain \textbf{words} are more likely to come after other words (e.g., the word \textit{dog} is somehow likely to come after the word \textit{barking} but perhaps less likely to come after the word \textit{fruitful}; this concept is known as the n-gram model in natural language processing~\cite{ngram}). 

Concretely, in this paper, for the calculation of transition probabilities used in the context of chord progressions, we took a data-based approach. In particular, we considered a data set consisting of 1,410 Jazz songs (iRealPro -- \url{https://www.irealpro.com/main-playlists}). To extract the chord symbols from this data set, we utilized the \url{https://github.com/pianosnake/ireal-reader}{ireal-reader} tool. After acquiring the chord symbols, our next step involved converting each chord into a simplified representation using the alphabet shown in Figure~\ref{figure:alphabet}. For the purpose of aggregation, we filtered 1015 songs from the data set (in particular, those that are precisely 32 bars in length, each encompassing a maximum of two chords per bar; technically, when encountering instances where only a single chord was present in a bar, we replicated it to ensure uniformity and consistency throughout the data set). 
Then, we used the adjusted data set to compute the probabilities of two chords being successive in an harmonic sequence.

As a result, we have a complete, directed graph with weighted arcs that represent the probability of one chord appearing after another, defined as $G = (V, w)$. In this context, each vertex $v \in V$ corresponds to a chord from "chords alphabet" as shown in Figure~\ref{figure:alphabet}. The set of edges, $E$ encompasses all possible arcs, where the weight $w_{(u, v)}$ on the arc $(u, v) \in E$ signifies the degree of smoothness associated with transitioning from chord $u$ to chord $v$, with $0 \leq w_{(u, v)} \leq 1$.

\section{Formal Model}\label{section:formal model}

We describe our formal model and discuss how to evaluate the quality of different winning harmonizations.

\paragraph{A Formal Model of Collaborative Harmonization}
With respect to a certain "harmonic alphabet $A$ (containing the set of $m$ possible chords), an \emph{instance} of our model contains the following ingredients:
\begin{itemize}

\item
A given \emph{size} $k$ of the harmonic sequence to produce.

\item
A set $V = \{v_1, \ldots, v_n\}$ of $n$ agents; each agent suggests its ideal harmonic sequence, denoted by $v_i$, where $v_i \in A^k$.

\end{itemize}

It is thus convenient to denote an instance of the model by $(k, V)$.
Given an instance $(k, V)$ of the model, a \emph{solution} $W$ corresponds to an aggregated harmonic sequence; formally, $W \in A^k$.

An \emph{aggregation method} (i.e., a voting rule) for the setting of collaborative harmonization takes an instance $(k, V)$ as its input and outputs a solution $W$.

\begin{example}
Figure~\ref{fig:aggregation example} is an example of an input and an output of our algorithms; the example contains $5$ agents (depicted on the left) and a possible aggregated harmonic sequence (depicted on the right).
\begin{figure*}[t]
    \centering
    \includegraphics[width=15cm]{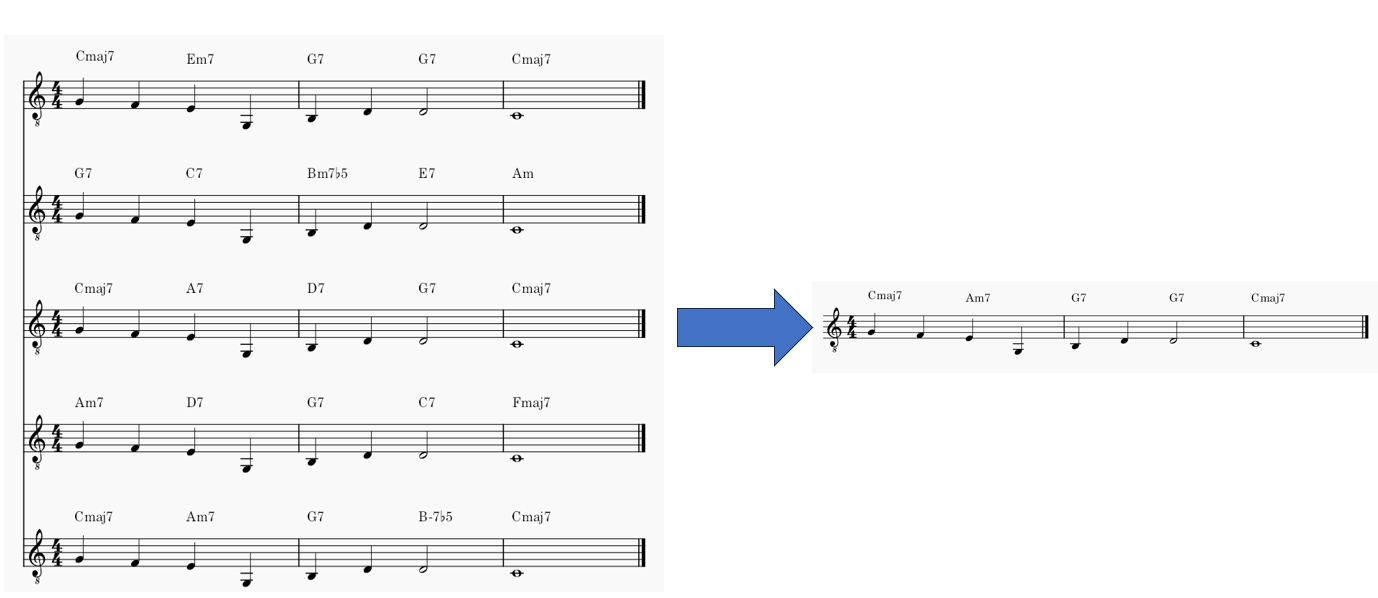}
    \caption{An example of an instance to the collaborative harmonization problem (on the left) with one possible aggregated harmonic sequence (on the right).}
    \label{fig:aggregation example}
\end{figure*} 
\end{example}

\begin{remark}
Note that, perhaps surprisingly, the melody itself is not part of instances of our model, as we address only the harmonization; however, refer to Section~\nameref{section:outlook} to a more elaborate discussion regarding this point.
\end{remark}

\paragraph{Solution Quality}
What is missing from the section above is a discussion on how to evaluate the quality of solutions (i.e., aggregated harmonic sequences) of instances of our model.
Informally speaking, we are looking for aggregation methods whose output strike a balance between these two aspects:
(1) first, a solution shall respect the suggestions of the agents -- i.e., the agents ideal harmonic sequences shall be taken into account by the aggregation algorithm;
and (2) second, a solution should be musically appealing on its own.

There are several ways to formally capture these two desires. Below we describe our approach, which builds on (1) chord similarity for the first aspect; and (2) chord progression probabilities for the second aspect.

\subsection{Agent Satisfaction}

Our approach at the consideration of the correspondence between the ideal harmonic sequences of the agents and a possible winning harmonization is to define \emph{agent satisfaction} (this follows the utility-based approach in social choice, such as used, e.g., for participatory budgeting~\cite{talmon2019framework}).
Concretely, consider some agent $v_i$ and a possible winning harmonization $W$; intuitively, the more similar $v_i$ is to $W$ (with respect to the metric used for chord similarity) the more satisfied $v_i$ shall be. 
Formally, we define as follows:
\begin{itemize}

\item
Let $d$ be the Jaccard distance (refer to Section Similarity). 

\item
We then define, for agent $i$ voting as $v_i$ and a possible winning harmonization $W$, the satisfaction of agent $i$ from $W$ to be (assume that $|v_i| = |W| = k$:
$$sat(i, W) :=  \sum_{j \in [k]} d(v_i[j], W[j])\ .$$

\end{itemize}

Using the concept of agent satisfaction, one may, e.g., consider the goal of maximizing the social welfare (i.e., maximizing the sum of agent satisfaction $\Sigma_{i \in [n]} sat(i, W)$).
In this paper we largely take this utilitarian approach.

\subsection{Musical Coherence}\label{section:musical coherence}

Besides corresponding to the ideal harmonic sequences of the agents, a solution of a collaborative harmonization instance shall also be musically coherent.
Our approach towards the assessment of the musical quality of a potential winning harmonization builds on the 2-gram approach described in Section~\nameref{section:preliminaries}.
While, indeed, musical coherence is more involved than our simple 2-gram approach (such as using an n-gram), our approach nevertheless captures the basics of chord progression: acoustically, a smoother transition sequence between chords contributes to a smoother ``flow''. 

\begin{remark}
While it makes sense to consider Musical Coherence of the whole harmonization as a whole, a natural simplification is to consider pairwise transitions.
\end{remark}

Formally, we define Musical Coherence as follows:
\begin{itemize}

\item We start by the \textit{chord transition graph} $G$, as described in Section~\nameref{section:preliminaries}.

\item
Next, we define the \emph{Musical Coherence} of a potential winning harmonization, denoted as $W$. It is calculated as the product of consecutive smoothness values represented by the weights of $G$, denoted by $w$. Formally, the Musical Coherence of $W$ can be expressed as:
$$sat(W) := \sum_{j \in [k - 1]} \log(w(W[j], W[j + 1]))\ .$$

\end{itemize}

This mathematical framework allows us to quantitatively evaluate the musical quality of a harmonization by considering the smoothness of transitions between chords in the context of the defined chord transition graph.

\subsection{Example and Discussion}

Consider the following collaborative harmonization scenario involving a sequence of $4$ chords and a community of $3$ agents.
The ideal harmonic sequence of the three agents are as follows:
\begin{itemize}

\item
Agent 1: Cmaj7, Dm7, Db7, Cmaj7.

\item
Agent 2: Am7, Dm7, E7, Am7.

\item
Agent 3: Cmaj7, Fmaj7, G7, Am7.

\end{itemize}


 \begin{table}[t]
\centering
 \setlength{\tabcolsep}{12pt}
 \begin{tabular}{|c||c|}
 \hline
 Chord & Notes \\
 \hline\hline
 Cmaj7 & C,E,G,B \\
 \hline
 Am7& A,C,E,G\\
 \hline
 Dm7& D,F,A,C\\
 \hline
 Db7& Db,F,Ab,B\\
 \hline
 Fmaj7& F,A,C,E\\
 \hline
 E7& E,Ab,B,D\\
 \hline
 G7& G,B,D,F\\
 \hline
 \end{tabular}
 \caption{Chords and Notes they Contain.}
 \label{table:Chords and Notes they contain}
 \end{table}

Suppose we have a potential solution, as follows
\begin{enumerate}
    \item  Cmaj7, Dm7, E7, Am7.
\end{enumerate}

 Agent satisfaction (based on Jaccard distance of common notes of two chords) is calculated using table~\ref{table:Chords and Notes they contain} as follows:
\begin{align*}
\text{Agent 1:} &\ (0 + 0 + \frac{2}{3}  + 0.4) =1.0667 \\
\text{Agent 2:} &\ (0.4 + 0 + 0 + 0) = 0.4 \\
\text{Agent 3:} &\ (0 + 0.4 + \frac{2}{3} + 0) = 1.0667 \\
\end{align*}
Thus the total satisfaction is $2.53334$.

\section{Approaches and Algorithms}\label{section:aaa}

In this section, we develop and discuss algorithms for solving instances of collaborative harmonization.S
For convenience, we will use the following notation:
\begin{itemize}

\item 
The preferences of the agents are denoted using a matrix $\mathcal{B}$ of size $n \times k$ (for $n$ agents and $k$ being the length of the harmonic sequence). Each element $b_{i,j}$ in this matrix represents the chord selected by agent $i$ in position $j$.

\item
Given such an instance, a solution is denoted by $W \in A^k$.

\end{itemize}

\paragraph{Plurality}\label{section:plurality}
We first consider the adaptation of the Plurality rule to our setting. In this simple aggregation method, we consider each chord-position independently; and, seek to find the most popular chord for every position of the $k$ chords. Procedurally, Plurality works by calculating a score $M(W)$ for each chord $W[j]$ as follows:
\begin{align}
    M(W) = \sum_{i=1}^n\sum_{j=1}^k \mathbb{I}(b_{i,j}=W[j])\ .
\end{align}

In this equation, $b_{i,j}$ represents the chord selected by agent $i$ in position $j$ of matrix $B$, and $\mathbb{I}$ is the indicator function.
And, the output of the Plurality Algorithm is the chord $W$ that maximizes $M(W)$, represented as $\arg\max_W{(M(W))}$.

\paragraph{Kemeny Rule}\label{section:kemeny}
The Kemeny rule~\cite{ali2012experiments} is a well-known aggregation method that is applicable to many social choice settings. In the standard model of ordinal-based social choice, Kemeny proceeds by assigning a score to each possible social welfare function (i.e., a linear order of the available candidates) on the distance from that permutation to all other individual permutations (i.e., votes); and selecting the one for which this sum of (swap) distances is the smallest.
We apply this concept to our setting. This is accomplished by using the Jaccard distance function described in Section~\nameref{section:preliminaries}. Correspondingly, in our adaptation of Kemeny, at each position within the chord sequence $W \in A^k$, we wish to select a chord $W[j]$ for position $j$ in a way that minimizes the cumulative distance from all individual chords $b_{i,j}$.
Formally, we define the Kemeny optimization quantity as follows:
\begin{align}
    K(W) = \left( \sum_{i=1}^n \sum_{j=1}^k d(b_{i,j}, W[j]) \right)\ .
\end{align}

Here, the function $d$ quantifies the dissimilarity between two chords. In this paper, we consider $d$ as the Jaccard distance, and so $d(b_{i,j}, W[j])\in [0,1]$ for all $i,j$.
The output of the algorithm is $\arg\min_{W}({K(W)})$.

\paragraph{PAV}\label{section:pav}
Next, we are interested in proportionality~\cite{skowron2016}; and, to this end, we adapt the PAV voting rule to our setting.
Concretely, we consider the \emph{proportional Borda count} rule: this rule utilizes the harmonic series to assign scores to candidates or preferences based on their rankings, aiming for proportional representation according to their ranking~\cite{dey2017proportional}. In our context, the objective is to create a chord sequence that effectively represents the diverse preferences of the agents, proportionally. To accomplish this, we employ an objective functions that is largely similar to the proportional Borda count, albeit where the Borda score is replaced by our Jaccard metric.
Formally, the objective function is defined as follows:
\begin{align}
P(W) = \left( \sum_{i=1}^n \sum_{j=1}^k \frac{1}{j} \cdot \left( s\left(U\left(B_{i},W\right)\right)[j]\right) \right)\ .
\end{align}

In this equation:
\begin{itemize}
    \item $U$ is a utility function that takes two $k$-sized vectors and returns a $k$-sized vector of utilities. We set $U(B_i,W)[j]=1-d(b_{i,j},W[j])$.
    
    \item $B_i$ represents a $k$-sized vector of chords, representing the $i$th row in matrix~$B$.
    
    \item $s$ denotes a sorting function that takes a $k$-sized vector and arranges it in descending order.
\end{itemize}

The output of the algorithm is $\arg\max_W({P(W)})$. As this problem is naturally NP-Hard (see Section~\nameref{section:complexity}), for its computation (in our computer-based simulations), we utilized a heuristic approach a local search algorithm of simulated annealing.

\paragraph{Clustered-Kemeny}\label{section:clustered kemeny}
We introduce Clustered-Kemeny, an algorithm based on the Kemeny voting rule, adapted for the natural division of musical pieces into sections. The goal is to identify clusters of individuals with similar chord preferences, enabling the partitioning of the musical piece into sections and matching each cluster to its most representative section.

Assuming a given partition of the harmonic sequence's length $k$ into $x \leq n$ continuous sections, we formulate a linear program for optimizing voter clustering and section matching. The partition, represented by vector $Z$, must satisfy specific conditions.

The problem is divided into two nested sub-problems:
(1) Given a partition $Z$, find the optimal clustering of agents into sections; and (2) Find the optimal $Z$ from all possible partitions into at most $X$ sections, returning the optimal agent clustering.
The objective is to maximize total satisfaction by minimizing the distance of selected chord solution $W$ and their respective sections within the clustering.
Formally:
\[
\min_{Z\in \text{partitions}}\min_{a_i\in Z, W} \sum_{z \in Z} \sum_{i=1}^n \sum_{j=1}^k P(j, z, Z) \cdot Q(i, z, Z) \cdot (d(b_{i, j}, W[j]))
\]

Where: (1) $Q(i, z, Z)$ indicates whether agent $i$ is in section $z$ (1 if true, otherwise a specified value less than 1); (2) $P(j, z, Z)$ is 1 if position $j$ belongs to section $z$, 0 otherwise; and (3) For the last section: $P(j, x, Z)$ is 1 if $Z[x] \leq j \leq k$, 0 otherwise.
These collectively define the optimization problem for the Clustered-Kemeny algorithm.

\subsection{2-Gram-Based Approaches}\label{label:two gram}

Note that, above, we have only relied on the Jaccard distance for quantifying the chord similarity (between the ideal harmonic sequence of each agent and some proposed solution). Next we consider also the musical coherence (or harmonic flow), corresponding to the 2-gram approach described in Section~\nameref{section:musical coherence}.

In particular, the following approaches build upon the previous methodologies and share the same input as their predecessors; but they aim to maximize a new objective function, denoted as $G$, which is defined as follows using transition probabilities $p$:
\begin{align}
G(W) = -\sum_{i=0}^{k-1}\left(\log\left(p(W[i], W[i+1])\right)\right)\ .
\end{align}

Here, the transition probabilities $p$ quantify the likelihood of transitioning from one chord, \( \text{chord}_i \), to another chord, \( \text{chord}_{i+1} \), within a sequence. These probabilities are derived from observed frequencies of such transitions in our data set as described in Section~\nameref{section:musical coherence}.

These 2-Grams Based Algorithms expand upon the previous approaches, considering transition probabilities to enhance predictive accuracy while accommodating the collaborative nature of chord sequences.
Below we formally describe the details of taking into account such transition probabilities in the computation of the specific aggregation algorithms described above.

\paragraph{Plurality with 2-gram}
The objective is given by
$argmax_W({x_M\cdot M(W) + (1-x_M)\cdot G(W)})$, where $x_M\in (0,\cdots,1)$ is a constant.

\paragraph{Kemeny with 2-gram}
The objective is determined by
$argmin_W({x_K\cdot K(W) - (1-x_K)\cdot G(W)})$, where $x_K\in (0,\cdots,1)$ is a constant.

\paragraph{PAV with 2-gram}
The objective is determined by
$argmin_W({x_P\cdot P(W) - (1-x_P)\cdot G(W)})$, where $x_P\in (0,\cdots,1)$ is a constant.

\paragraph{Clustered-Kemeny with 2-gram}
The objective is determined by
$argmin_W({x_{KC}\cdot KC(W) - (1-x_{KC})\cdot G(W)})$, where $x_{KC}\in (0,\cdots,1)$ is a constant.

\section{Computational Complexity}\label{section:complexity}

We present the computational complexity of the aggregation goals defined above in Table~\ref{complexity-table}. Additionally, we provide two sketches of hardness proofs. The complete proofs, along with all other missing proofs of the complexity results, are included in the supplementary material.
\begin{table}[t]
\centering
\setlength{\tabcolsep}{12pt}
\begin{tabular}{|c||c|}
\hline
\textbf{Problem} & \textbf{Complexity} \\
\hline\hline
Plurality & Poly-time\\
\hline
Kemeny& Poly-time\\
\hline
Proportional& NP-hard\\
\hline
Plurality with 2-gram& Poly-time\\
\hline
Kemeny with 2-gram& Poly-time\\
\hline
Proportional with 2-gram& NP-hard\\
\hline
Clustered-Kemeny & NP-hard\\
\hline
Clustered-Kemeny with 2-gram & NP-hard\\
\hline
\end{tabular}
\caption{Various problems of collaborative harmonization and their computational complexity.}
\label{complexity-table}
\end{table}
\begin{theorem}
  PAV and PAV with 2-gram are NP-hard.
\end{theorem}

\begin{proof}
We prove the NP-hardness of the Proportional algorithm by noting that it includes Proportional Approval Voting (PAV) as a special case, where \(U[i]\in \{0,1\}\). This implies Proportional is NP-hard~\cite{skowron2016finding}. By setting \(x_p = 1\), which reduces to the Proportional algorithm, we establish that Proportional with 2-gram is also NP-hard.
\end{proof}

\begin{theorem}
  Clustered-Kemeny and Clustered-Kemeny with 2-gram are NP-hard.
\end{theorem}

\begin{proof}
We reduce the k-median problem, known to be NP-Hard~\cite{charikar1999constant}, to Clustered-Kemeny. Given $n$ strings \(S=\{s_1, s_2, \ldots, s_n\}\) of length $\ell$, a distance metric \(d\), and a threshold \(t\), the k-string median problem seeks $k$ median strings \(M=\{m_1, m_2, \ldots, m_k\}\) such that:
\[
  \sum_{i=1}^n \min_{1 \leq j \leq k} d(s_i, m_j) \leq t.
\]

We construct a Clustered-Kemeny instance by replicating each input string $k$ times, transforming it into an agent representation (e.g., "abc" becomes "abcabcabc"). Partitions and lengths are defined as $Z=\{L, L \cdot 2, L \cdot 3, \ldots, L \cdot k\}$. Setting $Q(i,z,Z) = 0$ for each agent \(i\), we show that the k-median problem is a yes instance if and only if Clustered-Kemeny is a yes instance, establishing its NP-hardness.

For Clustered-Kemeny with 2-gram, setting \(x_{KC}=1\) aligns it with the hardness proof of Proportional with 2-gram, thus proving it NP-hard as well.
\end{proof}

\section{Computer-Based Simulations}\label{section:simulations}

Next, we report on computer-based simulations to evaluate the quality of the proposed algorithms. Given the NP-hard nature of the problems, we used a heuristic approach with a simulated annealing solution. The local search was initiated with the Plurality algorithm, and each search had 1000 iterations.

Due to the scarcity of relevant data for collaborative harmony composition, we adopted a semi-artificial approach. We used real harmonizations of songs and introduced random perturbations to simulate ideal harmonies by different agents.

For the dataset, we used 8, 16, and 32 agents, each contributing an ideal chord progression. Each agent's progression had 8, 16, and 32 variations, respectively, of the original progression of songs from a processed dataset of 1015 jazz songs (see section~\ref{Grammar}). Variations were introduced by random chord swaps within ranges (0, 1), (1, 2), (2, 3), and (3, 4). The new chord was selected based on Jaccard distances, ensuring musical coherence.

For the 2-gram-based algorithms, we fine-tuned the weights through several rounds of testing, arriving at optimal weights: $X_M = 0.5$, $X_K = 0.9$, $X_P = 1 - 2e^{-4}$, and $X_{KC} = 0.9$.

Evaluation metrics were based on three measures for any solution $W$:

\begin{itemize}
\item \textbf{Song Similarity Measure}: Assesses the proximity of the aggregated chord progression to the original sequence. Lower distances indicate higher adherence to the original progression:
\[
\sum_{j=1}^m d(W[j], \text{original song}[j])\ .
\]

\item \textbf{Cluster Coherence}: Evaluates the coherence of chords within each 16-bar section of the aggregated chord progression. Lower cluster distances indicate smoother transitions:
\[
\frac{1}{(m-16)\cdot n}\sum_{i=1}^{n}\sum_{j=1}^{m-16} \sum_{k=j}^{j+16} d(W[k], b_{i,k})\ .
\]

\item \textbf{Musical Coherence}: Measures the appropriateness of the aggregated chord progression within the context of the 2-gram. Higher scores indicate more harmonious sequences, as detailed in the supplementary material.
\end{itemize}

\section{Results and Discussion}

\begin{figure}[t]
  \centering
    \caption{Song Similarity vs Error Ranges for 32 Agents.}
  \includegraphics[width=0.5\textwidth]{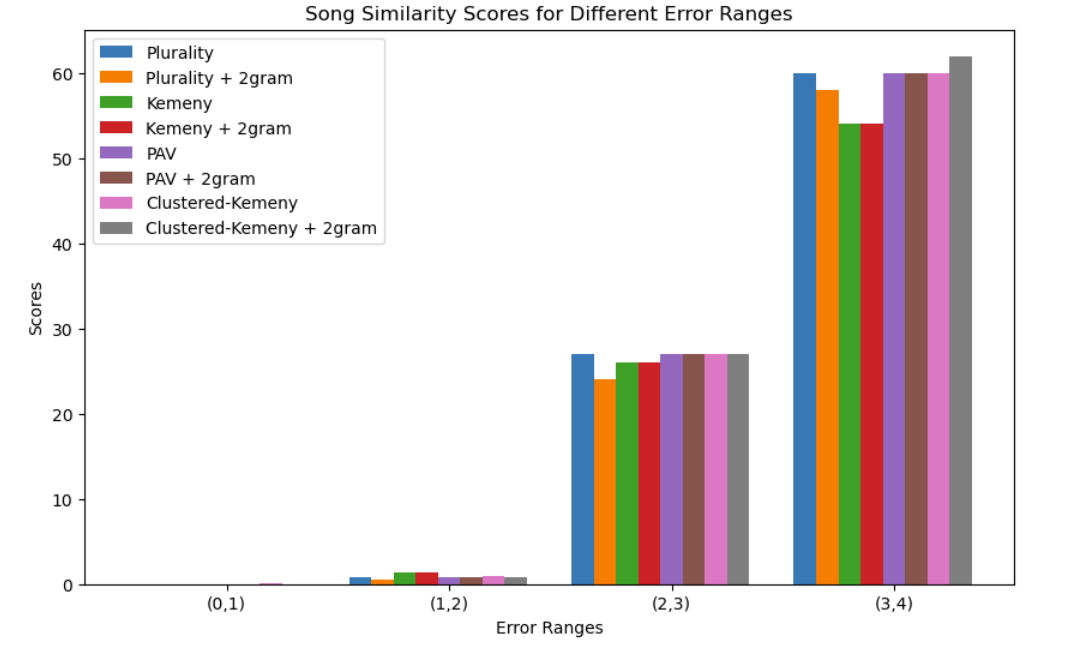}
  \label{fig:1}
\end{figure}

\begin{figure}[t]
  \centering
  \includegraphics[width=0.5\textwidth]{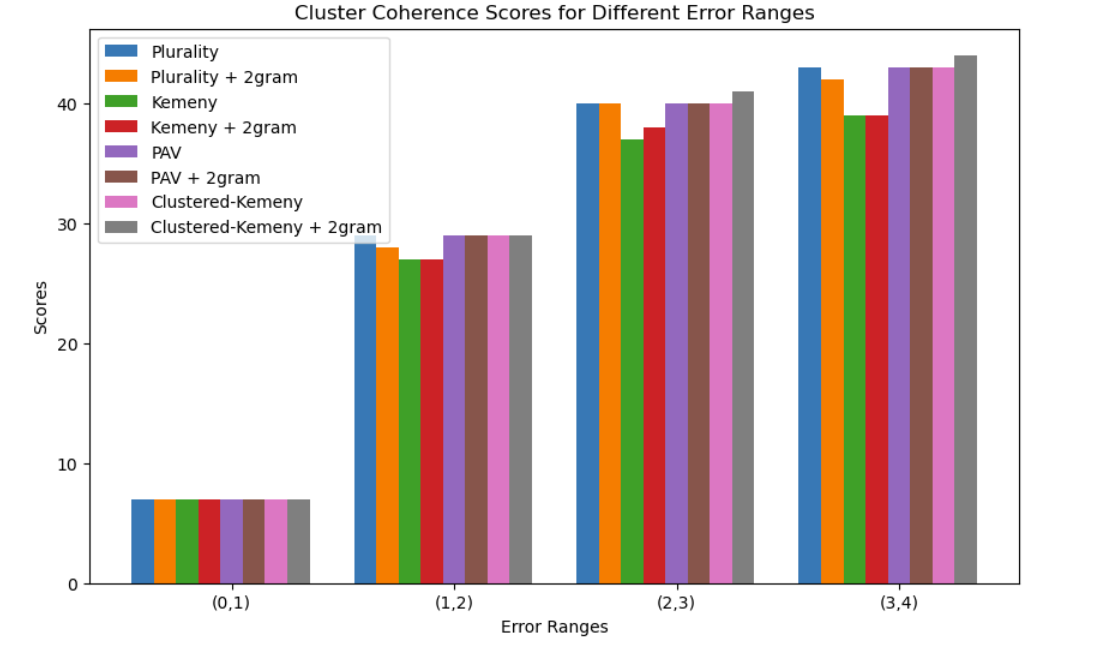}
  \caption{Cluster Coherence vs Error Ranges for 32 Agents.}
  \label{fig:2}
\end{figure}
\begin{figure}[t]
  \centering
  \includegraphics[width=0.5\textwidth]{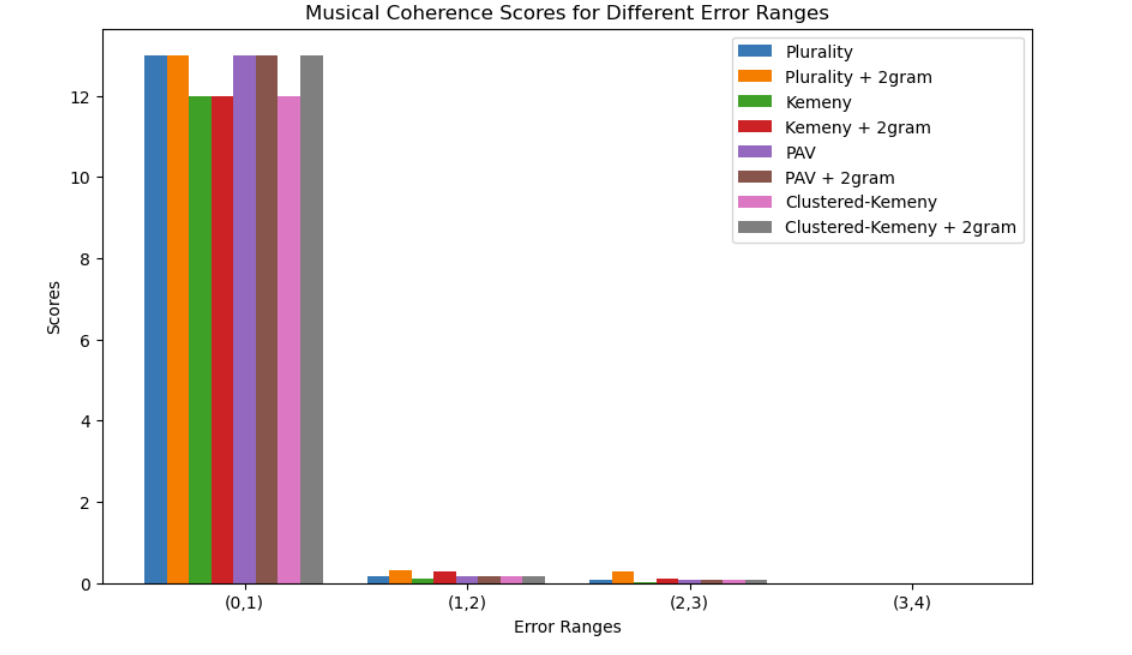}
  \caption{Musical Coherence vs Error Ranges for 32 Agents.}
  \label{fig:3}
\end{figure}

The full results shown in the supplementary material the average scores of all 1015 instances, each scaled by a factor of 100. These scores are associated with various parameters: a specific measurement, an algorithm used, the number of chord swaps employed to generate different variations, and the number of agents (representing the variations).

Figure \ref{fig:3} demonstrates a trend where an increase in the number of chord swaps (from 0 to 1) leads to a significant decline in musical coherence. This decline occurs because employing the Jaccard distance for chord swaps often generates chord progressions or variations that are less probable within the data set. Consequently, these less probable variations adversely affect the outcomes produced by the aggregation algorithms.

Further, direct analysis of the results are given in the supplementary material and the source code can be found here~\url{https://anonymous.4open.science/r/CollaborativeHarmonization/Probabilities%20and%20Distances.txt}.

\subsection{Insights and Implications}

Here are some of the main insights we derive from the results of the simulations:
\begin{itemize}
    \item Robustness to Errors- Plurality-based algorithms exhibit robustness to errors, maintaining better proximity to original sequences and higher Musical Coherence across varying error ranges and agent counts.
    \item Balancing Chord Adherence and Coherence- Kemeny methods strike a balance between chord adherence and musical coherence, particularly effective when integrated with 2-gram.
    \item Scalability Concerns- Plurality-based algorithms demonstrate scalability and consistent performance across different agent counts. In contrast, Clustered-Kemeny might face challenges in larger-scale settings, potentially indicating scalability issues.
    \item Clustering - It is intriguing to note that despite the anticipated ability of Clustered-Kemeny to generate coherent and adherent sections, it appears to have fallen short in achieving this objective. One plausible explanation could be attributed to the creation of numerous small clusters, each smaller than 16 bars, coupled with only a few agents assigned to each cluster. This scenario led to the formation of multiple 16-bar clusters that lacked coherence, consequently resulting in an inadequate evaluation of the measure's effectiveness.
\end{itemize}

\subsection{Summary}

Essentially, we observe that, while Kemeny + 2-gram doesn't consistently yield the best results across all combinations of measures, number of agents, and number of errors affecting the variations of a song, it reliably positions itself in a solid middle ground. It may not excel in every scenario, but it consistently maintains a respectable performance, offering a balance among different evaluation measures when considering diverse variations and conditions.

\section{Conclusions}\label{section:outlook}

In this study, our primary goal was to explore collaborative harmonization as a case study, drawing parallels to the structured nature of text aggregation. This comparative approach allowed us to gain insights by leveraging the structured framework inherent in harmonization, potentially informing strategies for text aggregation in the future. We established a formal model, introduced various algorithms, and found that simpler algorithms consistently performed better across all considered measures.

Moving forward, we suggest several areas for future work. Firstly, the chosen Jaccard distance metric presented in this work, is less practical and does not seem to capture the entire of essence of musical harmonization and thus investigating other, more relevant distance metrics is needed. Secondly Our research focused on pairwise transitions forming 2-grams for harmony suitability. A more refined approach involving $n$-grams would provide a more comprehensive analysis of chord relationships beyond consecutive pairs.

Additionally, our study exclusively addressed harmonization, neglecting melody considerations. Future investigations should integrate melody into the harmonization process to enhance musical coherence and quality.

An extension of our work could involve using more common practiced distance metrics and empirical assessments with real musicians, providing valuable insights into how these algorithms are perceived and valued in practical musical contexts.

All our algorithms share a vertical nature, comparing chords in the same position. Addressing the limitation of not recognizing equivalence in non-identical chord positions is a pertinent avenue for future research to improve contextual understanding.

For future endeavors in text aggregation, recognizing and utilizing structured string or text scenarios can greatly facilitate more meaningful aggregation. Aligning the aggregation process with the inherent structure present in the text can significantly enhance the quality and relevance of aggregated outcomes.

\bibliography{bib}

\begin{thebibliography}{23}
\providecommand{\natexlab}[1]{#1}
\providecommand{\url}[1]{\texttt{#1}}
\expandafter\ifx\csname urlstyle\endcsname\relax
  \providecommand{\doi}[1]{doi: #1}\else
  \providecommand{\doi}{doi: \begingroup \urlstyle{rm}\Url}\fi

\bibitem[Ali and Meila(2012)]{ali2012experiments}
A.~Ali and M.~Meila.
\newblock Experiments with kemeny ranking: What works when?
\newblock \emph{Math. Soc. Sci.}, 64\penalty0 (1):\penalty0 28--40, 2012.
\newblock \doi{10.1016/J.MATHSOCSCI.2011.08.008}.
\newblock URL \url{https://doi.org/10.1016/j.mathsocsci.2011.08.008}.

\bibitem[Aziz and Shah(2020)]{aziz2021participatory}
H.~Aziz and N.~Shah.
\newblock Participatory budgeting: Models and approaches.
\newblock \emph{CoRR}, abs/2003.00606, 2020.
\newblock URL \url{https://arxiv.org/abs/2003.00606}.

\bibitem[Brandt et~al.(2016)Brandt, Conitzer, Endriss, Lang, and
  Procaccia]{arrow2010handbook}
F.~Brandt, V.~Conitzer, U.~Endriss, J.~Lang, and A.~D. Procaccia, editors.
\newblock \emph{Handbook of Computational Social Choice}.
\newblock Cambridge University Press, 2016.
\newblock ISBN 9781107446984.
\newblock \doi{10.1017/CBO9781107446984}.
\newblock URL \url{https://doi.org/10.1017/CBO9781107446984}.

\bibitem[Briman and Talmon(2023)]{briman2023multiple}
E.~Briman and N.~Talmon.
\newblock Multiple attribute list aggregation and an application to democratic
  playlist editing.
\newblock In V.~Malvone and A.~Murano, editors, \emph{Multi-Agent Systems -
  20th European Conference, {EUMAS} 2023, Naples, Italy, September 14-15, 2023,
  Proceedings}, volume 14282 of \emph{Lecture Notes in Computer Science}, pages
  1--16. Springer, 2023.
\newblock \doi{10.1007/978-3-031-43264-4\_1}.
\newblock URL \url{https://doi.org/10.1007/978-3-031-43264-4\_1}.

\bibitem[Bulteau et~al.(2021)Bulteau, Shahaf, Shapiro, and
  Talmon]{bulteau2018aggregation}
L.~Bulteau, G.~Shahaf, E.~Shapiro, and N.~Talmon.
\newblock Aggregation over metric spaces: Proposing and voting in elections,
  budgeting, and legislation.
\newblock \emph{J. Artif. Intell. Res.}, 70:\penalty0 1413--1439, 2021.
\newblock \doi{10.1613/JAIR.1.12388}.
\newblock URL \url{https://doi.org/10.1613/jair.1.12388}.

\bibitem[Cavnar et~al.(1994)Cavnar, Trenkle, et~al.]{ngram}
W.~B. Cavnar, J.~M. Trenkle, et~al.
\newblock N-gram-based text categorization.
\newblock In \emph{Proceedings of SDAIR-94, 3rd annual symposium on document
  analysis and information retrieval}, volume 161175, page~14. Las Vegas, NV,
  1994.

\bibitem[Charikar et~al.(2002)Charikar, Guha, Tardos, and
  Shmoys]{charikar1999constant}
M.~Charikar, S.~Guha, {\'{E}}.~Tardos, and D.~B. Shmoys.
\newblock A constant-factor approximation algorithm for the k-median problem.
\newblock \emph{J. Comput. Syst. Sci.}, 65\penalty0 (1):\penalty0 129--149,
  2002.
\newblock \doi{10.1006/JCSS.2002.1882}.
\newblock URL \url{https://doi.org/10.1006/jcss.2002.1882}.

\bibitem[Clement(1998)]{clement1998learning}
B.~J. Clement.
\newblock Learning harmonic progression using markov models, 1998.

\bibitem[Dey et~al.(2017)Dey, Talmon, and van Handel]{dey2017proportional}
P.~Dey, N.~Talmon, and O.~van Handel.
\newblock Proportional representation in vote streams.
\newblock In K.~Larson, M.~Winikoff, S.~Das, and E.~H. Durfee, editors,
  \emph{Proceedings of the 16th Conference on Autonomous Agents and MultiAgent
  Systems, {AAMAS} 2017, S{\~{a}}o Paulo, Brazil, May 8-12, 2017}, pages
  15--23. {ACM}, 2017.
\newblock URL \url{http://dl.acm.org/citation.cfm?id=3091134}.

\bibitem[Faliszewski et~al.(2023)Faliszewski, Flis, Peters, Pierczynski,
  Skowron, Stolicki, Szufa, and Talmon]{skowronparticipatory}
P.~Faliszewski, J.~Flis, D.~Peters, G.~Pierczynski, P.~Skowron, D.~Stolicki,
  S.~Szufa, and N.~Talmon.
\newblock Participatory budgeting: Data, tools, and analysis.
\newblock \emph{CoRR}, abs/2305.11035, 2023.
\newblock \doi{10.48550/ARXIV.2305.11035}.
\newblock URL \url{https://doi.org/10.48550/arXiv.2305.11035}.

\bibitem[Grossi and Pigozzi(2014)]{grossi2022judgment}
D.~Grossi and G.~Pigozzi.
\newblock \emph{Judgment Aggregation: {A} Primer}.
\newblock Synthesis Lectures on Artificial Intelligence and Machine Learning.
  Morgan {\&} Claypool Publishers, 2014.
\newblock ISBN 978-3-031-00440-7.
\newblock \doi{10.2200/S00559ED1V01Y201312AIM027}.
\newblock URL \url{https://doi.org/10.2200/S00559ED1V01Y201312AIM027}.

\bibitem[Hemaspaandra et~al.(2005)Hemaspaandra, Spakowski, and
  Vogel]{hemaspaandra2005complexity}
E.~Hemaspaandra, H.~Spakowski, and J.~Vogel.
\newblock The complexity of kemeny elections.
\newblock \emph{Theor. Comput. Sci.}, 349\penalty0 (3):\penalty0 382--391,
  2005.
\newblock \doi{10.1016/J.TCS.2005.08.031}.
\newblock URL \url{https://doi.org/10.1016/j.tcs.2005.08.031}.

\bibitem[Levine(2011)]{levine2011jazz}
M.~Levine.
\newblock \emph{The jazz theory book}.
\newblock " O'Reilly Media, Inc.", 2011.

\bibitem[Li et~al.(2015)Li, Vo, and Kowalczyk]{li2015aggregating}
M.~Li, Q.~B. Vo, and R.~Kowalczyk.
\newblock Aggregating multi-valued cp-nets: a csp-based approach.
\newblock \emph{J. Heuristics}, 21\penalty0 (1):\penalty0 107--140, 2015.
\newblock \doi{10.1007/S10732-014-9276-8}.
\newblock URL \url{https://doi.org/10.1007/s10732-014-9276-8}.

\bibitem[Paiement et~al.(2005)Paiement, Eck, and
  Bengio]{paiement2005probabilistic}
J.~Paiement, D.~Eck, and S.~Bengio.
\newblock A probabilistic model for chord progressions.
\newblock In \emph{{ISMIR} 2005, 6th International Conference on Music
  Information Retrieval, London, UK, 11-15 September 2005, Proceedings}, pages
  312--319, 2005.
\newblock URL \url{http://ismir2005.ismir.net/proceedings/1091.pdf}.

\bibitem[Persichetti(1961)]{harmonyBook}
V.~Persichetti.
\newblock Twentieth century harmony: creative aspects and practice, 1961.

\bibitem[Rohrmeier(2011)]{rohrmeier2011towards}
M.~Rohrmeier.
\newblock Towards a generative syntax of tonal harmony.
\newblock \emph{Journal of Mathematics and Music}, 5\penalty0 (1):\penalty0
  35--53, 2011.

\bibitem[Rohrmeier(2020)]{rohrmeier2020syntax}
M.~Rohrmeier.
\newblock The syntax of jazz harmony: Diatonic tonality, phrase structure, and
  form.
\newblock \emph{Music Theory and Analysis (MTA)}, 7\penalty0 (1):\penalty0
  1--63, 2020.

\bibitem[Skowron et~al.(2016)Skowron, Faliszewski, and
  Lang]{skowron2016finding}
P.~Skowron, P.~Faliszewski, and J.~Lang.
\newblock Finding a collective set of items: From proportional
  multirepresentation to group recommendation.
\newblock \emph{Artif. Intell.}, 241:\penalty0 191--216, 2016.
\newblock \doi{10.1016/J.ARTINT.2016.09.003}.
\newblock URL \url{https://doi.org/10.1016/j.artint.2016.09.003}.

\bibitem[Skowron et~al.(2017)Skowron, Lackner, Brill, Peters, and
  Elkind]{skowron2016}
P.~Skowron, M.~Lackner, M.~Brill, D.~Peters, and E.~Elkind.
\newblock Proportional rankings.
\newblock In C.~Sierra, editor, \emph{Proceedings of the Twenty-Sixth
  International Joint Conference on Artificial Intelligence, {IJCAI} 2017,
  Melbourne, Australia, August 19-25, 2017}, pages 409--415. ijcai.org, 2017.
\newblock \doi{10.24963/IJCAI.2017/58}.
\newblock URL \url{https://doi.org/10.24963/ijcai.2017/58}.

\bibitem[Talmon and Faliszewski(2019)]{talmon2019framework}
N.~Talmon and P.~Faliszewski.
\newblock A framework for approval-based budgeting methods.
\newblock In \emph{The Thirty-Third {AAAI} Conference on Artificial
  Intelligence, {AAAI} 2019, The Thirty-First Innovative Applications of
  Artificial Intelligence Conference, {IAAI} 2019, The Ninth {AAAI} Symposium
  on Educational Advances in Artificial Intelligence, {EAAI} 2019, Honolulu,
  Hawaii, USA, January 27 - February 1, 2019}, pages 2181--2188. {AAAI} Press,
  2019.
\newblock \doi{10.1609/AAAI.V33I01.33012181}.
\newblock URL \url{https://doi.org/10.1609/aaai.v33i01.33012181}.

\bibitem[Yi and Goldsmith(2007)]{yi2007automatic}
L.~Yi and J.~Goldsmith.
\newblock Automatic generation of four-part harmony.
\newblock In K.~B. Laskey, S.~M. Mahoney, and J.~Goldsmith, editors,
  \emph{Proceedings of the Fifth {UAI} Bayesian Modeling Applications Workshop
  {(UAI-AW} 2007), Vancouver, British Columbia, Canada, July 19, 2007}, volume
  268 of \emph{{CEUR} Workshop Proceedings}. CEUR-WS.org, 2007.
\newblock URL \url{https://ceur-ws.org/Vol-268/paper10.pdf}.

\bibitem[Zvi et~al.(2021)Zvi, Leizerovich, and Talmon]{zvi2021iterative}
G.~B. Zvi, E.~Leizerovich, and N.~Talmon.
\newblock Iterative deliberation via metric aggregation.
\newblock In D.~Fotakis and D.~R. Insua, editors, \emph{Algorithmic Decision
  Theory - 7th International Conference, {ADT} 2021, Toulouse, France, November
  3-5, 2021, Proceedings}, volume 13023 of \emph{Lecture Notes in Computer
  Science}, pages 162--176. Springer, 2021.
\newblock \doi{10.1007/978-3-030-87756-9\_11}.
\newblock URL \url{https://doi.org/10.1007/978-3-030-87756-9\_11}.

\end{thebibliography}
\section{Appendix}
\subsection{Complexity Proofs}\label{Complexity Proofs}

We begin with Plurality; then go on to Kemeny, PAV, and clustered Kemeny.

\subsubsection{Complexity of Plurality}

\begin{observation}
  Plurality can be computed in polynomial-time.
\end{observation}

\begin{proof}
This follows as Plurality treats each chord-position independently. Formally, the following algorithm has a polynomial time complexity:
  For every column \(B_j\), the algorithm selects the chord that appears the most.
\end{proof}

Also the 2-gram version of Plurality admits a polynomial-time algorithm.

\begin{theorem}
Plurality with 2-gram can be computed in polynomial time.
\end{theorem}

\begin{proof}
This proof follows a structure similar to the previous theorem, where we establish that \(T(1,a)\) is equivalent to Plurality without the 2-gram feature, making it a polynomial algorithm.
\end{proof}

\subsubsection{Complexity of Kemeny}

Next, we show that also Kemeny is polynomial-time solvable. This is perhaps surprising (as the original Kemeny voting rule that corresponds to minimizing the sum of swap distance in ordinal elections in NP-hard~\cite{hemaspaandra2005complexity}); it is so, however, as our adaptation of Kemeny also considers each chord-position separately.

\begin{observation}
Kemeny can be computed in polynomial time.
\end{observation}

\begin{proof}
We describe a polynomial-time algorithm:
  For every column \(B_j\) and for every chord, the algorithm selects the chord that minimizes the summation of distances considering all agents. This results in a complexity of \(O(m\cdot n\cdot k)\), which is polynomial.
\end{proof}

Furthermore, even Kemeny with 2-gram can be computed in polynomial time; here, however, we use dynamic programming that follows the chord progression.

\begin{theorem}
  Kemeny with 2-gram can be computed in polynomial time.
\end{theorem}

\begin{proof}
We establish the polynomial nature of Kemeny with 2-gram by describing an algorithm that is based on dynamic programming. Let $T(j,a)$ denote the minimal cost of Kemeny with 2-gram for the first $j$ chords, given that the last chord is $W[j]=a$, where $a$ can be any of the $m$ chords.

To begin, we note that $T(1,a)$ corresponds to the case where there is only one chord, essentially Kemeny without the 2-gram feature. This base case has been previously shown to be polynomial.

Next, we introduce a recurrence relation: $T(j,a) = \arg\min_{a'} [T(j-1,a') + \sum_{i=1}^n d(b_{i,j},a) - \log(p(a',a))]$. This recurrence efficiently computes the minimal cost for each chord $a$ for the first $j$ chords and thus, also for the total of k chords.

Since we can find $\arg\min_a T(k,a)$ in polynomial time, this algorithm's overall complexity is polynomial. This process is akin to finding the optimal solution $W$ in polynomial time.
\end{proof}
\subsubsection{Complexity of PAV}

\begin{theorem}
  PAV and PAV with 2-gram are NP-hard.
\end{theorem}

\begin{proof}
We establish the NP-hardness of the Proportional algorithm by considering the general utility function \(U\), which returns a vector of utilities for each agent \(i\) where \(U[j]\in [0,1]\) for all $j\in [k]$. We also note that the Proportional algorithm encompasses Proportional Approval Voting (PAV) as a special case, where \(U[i]\in \{0,1\}\). This establishes that Proportional is NP-hardness~\cite{skowron2016finding}.

By considering the special case of setting \(x_p\) to $1$, which is equivalent to the Proportional algorithm, we can conclude that Proportional with 2-gram is also  NP-hard.
\end{proof}

\subsubsection{Complexity of Clustered Kemeny}

We proceeed to consider Clustered-Kemeny.

\begin{theorem}
  Clustered-Kemeny and Clustered-Kemeny with 2-gram are NP-hard.
\end{theorem}

\begin{proof}
To establish the NP-hardness of Clustered-Kemeny and Clustered-Kemeny with 2-gram, we start by drawing a connection to the k-median problem, proven as NP-Hard~\cite{charikar1999constant}. In our case, we consider the k-string median problem since a chord sequence effectively represents a string. 
    
The k-string median problem is formally defined as follows: Given a set of $n$ strings $S=\{s_1, s_2, \ldots, s_n\}$ of length $\ell$, a string distance metric $d$ normalized to return a value between $0$ and $1$, and a threshold $t$, determine if there exist $k$ median strings $M=\{m_1, m_2, \ldots, m_k\}$ of length $\ell$ each such that:
\[
  \sum_{i=1}^n\min_{1 \leq j \leq k} d(s_i, m_j) \leq t\ .
\]
    
Given such an instance, we construct an instance to the decision variant of Clustered-Kemeny that is given the partitioning into sections, as well as defining lengths for each section is formally represented by:
\[
\sum_{z \in Z} \sum_{i=1}^n \sum_{j=1}^k P(j, z, Z) \cdot Q(i, z, Z) \cdot (d(b_{i, j}, W[j])) < t\ .
\]

To construct an instance for the k-median problem comprising $n$ strings of length $\ell$, we create a corresponding instance for the Kemeny-Clustering problem by transforming each input string into an agent representation. In this transformation, each string is replicated $k$ times to generate the agents for the Kemeny-Clustering instance (e.g., an input string "abc" replicated $k=3$ times transforms into an agent with the string "abcabcabc"). We then set $value=0$ for each agent $i$ in $Q(i,z,Z)$, and a partition into $k$ sections by indices is defined as $Z=\{L, L\cdot 2, L\cdot 3, \ldots, L\cdot k\}$. 

It is evident that the k-median problem is a yes instance if and only if the Kemeny-Clustering problem is a yes instance. This implies that even when the given partition and length are specified, the Clustered-Kemeny problem remains challenging, as it necessitates checking multiple possible partitions for each number of sections, thereby intensifying the complexity of optimization.

Furthermore, when considering Clustered-Kemeny with 2-gram, the methodology is similar to the hardness proof of Proportional with 2-gram, achieved by setting $x_{KC}=1$.
\end{proof}
\subsection{Simulations Full Result}\label{Simulations Full Result}

Next We provide the full results tables below.

\begin{table*}[ht]
\centering
\caption{Errors (0,1) - (3,4) and 8 agents - Algorithms Scores.}
\resizebox{\textwidth}{!}{%
\begin{tabular}{|c|c|c|c|c|}
\hline
Error & Algorithm & Song Similarity & Cluster Coherence & Musical Coherence \\
\hline
(0,1) & Plurality & 0.41 & 7.1 & 12 \\
(0,1) & Plurality + 2-gram & 0.31 & 7.1 & 13 \\
(0,1) & Kemeny & 0.46 & 7.1 & 11 \\
(0,1) & Kemeny + 2-gram & 0.39 & 7.1 & 11 \\
(0,1) & PAV & 0.41 & 7.1 & 12 \\
(0,1) & PAV + 2-gram & 0.41 & 7.1 & 12 \\
(0,1) & Clustered-Kemeny & 1 & 7.3 & 11 \\
(0,1) & Clustered-Kemeny + 2-gram & 0.073 & 7 & 13 \\
\hline
(1,2) & Plurality & 26 & 29 & 0.18 \\
(1,2) & Plurality + 2-gram & 23 & 28 & 0.31 \\
(1,2) & Kemeny & 22 & 27 & 0.12 \\
(1,2) & Kemeny + 2-gram & 21 & 27 & 0.3 \\
(1,2) & PAV & 26 & 29 & 0.18 \\
(1,2) & PAV + 2-gram & 26 & 29 & 0.18 \\
(1,2) & Clustered-Kemeny & 26 & 29 & 0.17 \\
(1,2) & Clustered-Kemeny + 2-gram & 26 & 29 & 0.18 \\
\hline
(2,3) & Plurality & 58 & 40 & 0.0025 \\
(2,3) & Plurality + 2-gram & 57 & 40 & 0.0066 \\
(2,3) & Kemeny & 52 & 37 & 0.0012 \\
(2,3) & Kemeny + 2-gram & 52 & 38 & 0.0052 \\
(2,3) & PAV & 58 & 40 & 0.0025 \\
(2,3) & PAV + 2-gram & 58 & 40 & 0.0025 \\
(2,3) & Clustered-Kemeny & 58 & 40 & 0.0024 \\
(2,3) & Clustered-Kemeny + 2-gram & 59 & 41 & 0.0031 \\
\hline
(3,4) & Plurality & 70 & 43 & 0.00083 \\
(3,4) & Plurality + 2-gram & 70 & 42 & 0.0024 \\
(3,4) & Kemeny & 67 & 39 & 0.00042 \\
(3,4) & Kemeny + 2-gram & 66 & 39 & 0.0018 \\
(3,4) & PAV & 70 & 43 & 0.00083 \\
(3,4) & PAV + 2-gram & 70 & 43 & 0.00085 \\
(3,4) & Clustered-Kemeny & 70 & 43 & 0.00082 \\
(3,4) & Clustered-Kemeny + 2-gram & 71 & 44 & 0.0011 \\
\hline
\end{tabular}}
\end{table*}

\begin{table*}[ht]
\centering
\caption{Errors (0,1) - (3,4) and 16 agents - Algorithms Scores.}
\resizebox{\textwidth}{!}{%
\begin{tabular}{|c|c|c|c|c|}
\hline
Error & Algorithm & Song Similarity & Cluster Coherence & Musical Coherence \\
\hline
(0,1) & Plurality & 0.0012 & 7 & 13 \\
(0,1) & Plurality + 2-gram & 0.0015 & 7 & 13 \\
(0,1) & Kemeny & 0.004 & 7 & 12 \\
(0,1) & Kemeny + 2-gram & 0.0029 & 7 & 12 \\
(0,1) & PAV & 0.0012 & 7 & 13 \\
(0,1) & PAV + 2-gram & 0.0012 & 7 & 13 \\
(0,1) & Clustered-Kemeny & 0.36 & 7.1 & 12 \\
(0,1) & Clustered-Kemeny + 2-gram & 0.031 & 7 & 13 \\
\hline
(1,2) & Plurality & 8.9 & 7.2 & 7.2 \\
(1,2) & Plurality + 2-gram & 7.2 & 7.2 & 7.1 \\
(1,2) & Kemeny & 8.8 & 8.2 & 8.2 \\
(1,2) & Kemeny + 2-gram & 8.2 & 8.2 & 8.1 \\
(1,2) & PAV & 8.9 & 8.9 & 8.9 \\
(1,2) & PAV + 2-gram & 8.9 & 8.9 & 8.9 \\
(1,2) & Clustered-Kemeny & 9.1 & 8.9 & 8.9 \\
(1,2) & Clustered-Kemeny + 2-gram & 8.9 & 8.9 & 8.9 \\
\hline
(2,3) & Plurality & 44 & 42 & 42 \\
(2,3) & Plurality + 2-gram & 42 & 42 & 42 \\
(2,3) & Kemeny & 41 & 40 & 40 \\
(2,3) & Kemeny + 2-gram & 40 & 40 & 40 \\
(2,3) & PAV & 44 & 44 & 44 \\
(2,3) & PAV + 2-gram & 44 & 44 & 44 \\
(2,3) & Clustered-Kemeny & 44 & 45 & 45 \\
(2,3) & Clustered-Kemeny + 2-gram & 45 & 45 & 45 \\
\hline
(3,4) & Plurality & 66 & 66 & 66 \\
(3,4) & Plurality + 2-gram & 66 & 66 & 66 \\
(3,4) & Kemeny & 62 & 62 & 62 \\
(3,4) & Kemeny + 2-gram & 62 & 62 & 62 \\
(3,4) & PAV & 66 & 66 & 66 \\
(3,4) & PAV + 2-gram & 66 & 66 & 66 \\
(3,4) & Clustered-Kemeny & 66 & 67 & 67 \\
(3,4) & Clustered-Kemeny + 2-gram & 67 & 67 & 67 \\
\hline
\end{tabular}}
\end{table*}

\begin{table*}[ht]
\centering
\caption{Errors (0,1) - (3,4) and 32 agents - Algorithms Scores with 2-grams included.}
\resizebox{\textwidth}{!}{%
\begin{tabular}{|c|c|c|c|c|}
\hline
Error & Algorithm & Song Similarity & Cluster Coherence & Musical Coherence \\
\hline
(0,1) & Plurality & 0 & 7 & 13 \\
(0,1) & Plurality + 2-gram & 0 & 7 & 13 \\
(0,1) & Kemeny & 0.0012 & 7 & 12 \\
(0,1) & Kemeny + 2-gram & 0.0009 & 7 & 12 \\
(0,1) & PAV & 0 & 7 & 13 \\
(0,1) & PAV + 2-gram & 0 & 7 & 13 \\
(0,1) & Clustered-Kemeny & 0.07 & 7.1 & 12 \\
(0,1) & Clustered-Kemeny + 2-gram & 0.015 & 7 & 13 \\
\hline
(1,2) & Plurality & 6.2 & 7.1 & 7.1 \\
(1,2) & Plurality + 2-gram & 5.1 & 7.1 & 7.1 \\
(1,2) & Kemeny & 6.1 & 7.2 & 7.2 \\
(1,2) & Kemeny + 2-gram & 5.8 & 7.2 & 7.2 \\
(1,2) & PAV & 6.2 & 7.1 & 7.1 \\
(1,2) & PAV + 2-gram & 6.2 & 7.1 & 7.1 \\
(1,2) & Clustered-Kemeny & 6.3 & 7.3 & 7.3 \\
(1,2) & Clustered-Kemeny + 2-gram & 6.1 & 7.1 & 7.1 \\
\hline
(2,3) & Plurality & 38 & 40 & 40 \\
(2,3) & Plurality + 2-gram & 37 & 40 & 40 \\
(2,3) & Kemeny & 35 & 38 & 38 \\
(2,3) & Kemeny + 2-gram & 35 & 38 & 38 \\
(2,3) & PAV & 38 & 40 & 40 \\
(2,3) & PAV + 2-gram & 38 & 40 & 40 \\
(2,3) & Clustered-Kemeny & 39 & 40 & 40 \\
(2,3) & Clustered-Kemeny + 2-gram & 38 & 40 & 40 \\
\hline
(3,4) & Plurality & 59 & 60 & 60 \\
(3,4) & Plurality + 2-gram & 60 & 60 & 60 \\
(3,4) & Kemeny & 54 & 55 & 55 \\
(3,4) & Kemeny + 2-gram & 55 & 55 & 55 \\
(3,4) & PAV & 59 & 60 & 60 \\
(3,4) & PAV + 2-gram & 59 & 60 & 60 \\
(3,4) & Clustered-Kemeny & 60 & 61 & 61 \\
(3,4) & Clustered-Kemeny + 2-gram & 60 & 61 & 61 \\
\hline
\end{tabular}}
\end{table*}

\subsection{Further Analysis of the Simulation Results}

We provide some further analysis of the simulation results.

\subsubsection{Distance Evaluation}

\begin{itemize}
    \item \textbf{Plurality and Plurality + 2-gram} consistently exhibit significantly lower Jaccard Distances, indicating their ability to maintain closer proximity to the original sequences, especially evident as the errors increase.
    \item \textbf{Kemeny and Kemeny + 2-gram} perform competitively, showcasing slightly higher Jaccard Distances compared to the Plurality-based algorithms but demonstrating resilience against errors in chord alterations.
    \item \textbf{PAV methods} generally showcase stable but relatively higher Jaccard Distances compared to Plurality and Kemeny algorithms, implying slightly inferior performance in maintaining sequence proximity.
\end{itemize}

\subsubsection{Cluster Coherence Analysis}

\begin{itemize}
    \item \textbf{Plurality, Plurality + 2-gram} consistently maintain relatively low cluster distances across different error ranges and agent counts, indicating their ability to preserve coherence and smooth transitions within musical segments.
    \item \textbf{Clustered-Kemeny}, especially in scenarios with higher errors and agent counts, shows higher cluster distances, implying challenges in maintaining cohesion within smaller musical sections compared to other algorithms.
\end{itemize}

\subsubsection{Musical Coherence}

\begin{itemize}
    \item \textbf{Plurality and Kemeny algorithms} generally showcase better Musical Coherence, indicating that the chord progressions they generate align more harmoniously with the provided 2-gram.
    \item \textbf{Plurality + 2-gram and Kemeny + 2-gram} approaches display competitive Musical Coherence, suggesting their effectiveness in creating chord progressions that fit well within the context of the provided 2-gram.
\end{itemize}

\subsection{Toy Example}

We provide a toy example for illustrating instances of the model and the aggregation algorithms used throughout the paper.

\begin{example}
    Let us consider a simple scenario where 3 agents were tasked to create a sequence of 4 chords each:

    \begin{itemize}
        \item Agent 1: Cmaj7, Dm7, G7, Cmaj7.
        \item Agent 2: Am7, Dm7, E7, Am7.
        \item Agent 3: Cmaj7, Fmaj7, G7, Am7.
    \end{itemize}

    To illustrate collaborative harmonization algorithms, let's assess a proposed solution: 
    
    $W$: Cmaj7, Dm7, G7, Am7

    \begin{itemize}
        \item Plurality: This solution stands out as optimal since each chord was chosen by 2 out of the 3 agents.
        \item Kemeny: The given solution also appears optimal, demonstrating a minimal Jaccard distance of $2.167$.
        \item PAV (Proportional Approval Voting): By computing PAV scores for each agent, we find:
              Agent 1: $0.5$, Agent 2: $0.9166$, Agent 3: $0.5$. The overall PAV score sums up to $1.91666$.
        \item Clustered Kemeny: Considering a maximum of 3 sections and setting $value=0$ (where each agent is clustered into one section and the rest of their chord sequence is disregarded):
              \begin{itemize}
                  \item With 1 partition, the Kemeny score remains $2.167$.
                  \item With 2 partitions, viable partitions include:
                        \begin{itemize}
                            \item First section: (Cmaj7, Dm7), Second section: (G7, Am7).
                            \item First section: (Cmaj7), Second section: (Dm7, G7, Am7).
                            \item First section: (Cmaj7, Dm7, G7), Second section: (Am7).
                            \item First section: (Cmaj7), Second section: (Am7, Dm7, G7).
                        \end{itemize}
                  \item With 3 partitions, feasible partitions include:
                        \begin{itemize}
                            \item First section: (Cmaj7, Dm7), Second section: (G7), Third section (Am7).
                            \item First section (Cmaj7), Second section (Dm7, G7), Third section (Am7).
                            \item First section (Cmaj7), Second section (Dm7), Third section (G7, Am7).
                        \end{itemize}
                  The minimal distance score of $0$  emerges with the partition of:
                  First section (Cmaj7), Second section (Dm7, G7), Third section: (Am7). The clustering of agents would be Agent 3 in the first section, Agent 1 in the second section, and Agent 2 in the third section.
              \end{itemize}
\item In algorithms that manipulate the 2-gram score $G(W)$ through addition or reduction, we compute $G(W)$ utilizing the 2-gram probabilities:

  $G(W)=-(\log(CMaj7 -> Dmin7)+\log(Dmin7 -> G7)+\log(G7 -> A-7))=
  -(\log(0.0252903)+(0.199777)+\log(0.0053198))=10.524207$ 
    \end{itemize}
\end{example}

\end{document}